\documentclass[conference]{IEEEtran}

\usepackage{amsmath, amsthm, amssymb}
\usepackage{graphicx}
\usepackage{verbatim}
\usepackage{mathrsfs}
\usepackage{algorithmic}
\usepackage{float}
\usepackage[lofdepth, lotdepth]{subfig}
\usepackage{url}
\usepackage{array}

\usepackage[usenames,dvipsnames]{pstricks}
\usepackage{epsfig}
\usepackage{pst-grad} 
\usepackage{pst-plot} 

\theoremstyle{plain}
\newtheorem{theorem}{Theorem}
\newtheorem{lemma}[theorem]{Lemma}

\theoremstyle{definition}

\newcommand{\B}{{\mathcal B}}

\newcommand{\M}{{\mathcal M}}

\newcommand{\bA}{{\boldsymbol A}}

\newcommand{\bu}{{\boldsymbol u}}
\newcommand{\bv}{{\boldsymbol v}}

\newcommand{\bz}{{\boldsymbol{z}}}

\newcommand{\fq}{\mathbb{F}_q}

\newcommand{\al}{\alpha}

\newcommand{\weight}{{\mathsf{wt}}}

\newcommand{\nin}{\noindent}

\newcommand{\seq}{\subseteq}

\newcommand{\et}{{\emph{et al.}}}

\newcommand{\pa}{(P,\boldsymbol{A})}
\newcommand{\pas}{(P^*,\boldsymbol{A}^*)}

\newcommand{\dt}{{\mathsf d}_{\sum}}
\newcommand{\dla}{{\mathsf{d}}_{\boldsymbol{A}}}
\newcommand{\dlas}{{\mathsf{d}}_{\boldsymbol{A}^*}}
\newcommand{\nma}{\mathcal{N}(\mathcal{M}, \boldsymbol{A})}
\newcommand{\nmas}{\mathcal{N}(\mathcal{M},\boldsymbol{A}^*)}
\newcommand{\nhj}{(n,\{H_j\}^k_1)}
\newcommand{\bas}{\boldsymbol{A}^*}

\begin{document}
\pagestyle{plain}

\title{Delay Minimization in Varying-Bandwidth Direct Multicast with Side Information}

 \author{
   \IEEEauthorblockN{
     Son Hoang Dau\IEEEauthorrefmark{1},
     Zheng Dong\IEEEauthorrefmark{2}, 
     Chau Yuen\IEEEauthorrefmark{3}
		} 
   \IEEEauthorblockA{
   Singapore University of Technology and Design\\     
		Emails: $\{${\it\IEEEauthorrefmark{1}sonhoang\_dau, 
		\IEEEauthorrefmark{2}dong\_zheng,
		\IEEEauthorrefmark{3}yuenchau}$\}$@sutd.edu.sg		
		}
	\and
  \IEEEauthorblockN{Terence H. Chan}
  \IEEEauthorblockA{
	Institute for Telecommunications Research\\
	University of South Australia\\
	Email: terence.chan@unisa.edu.au
	}
 }
\maketitle

\begin{abstract}
We study the delay minimization in a direct multicast 
communication scheme where a base station 
wishes to transmit a set of original packets to a group of clients.
Each of the clients already has in its cache a subset of the original
packets, and requests for all the remaining packets. 
The base station communicates \emph{directly} with the clients by broadcasting 
information to them. 
Assume that bandwidths vary between the station and different
clients. We propose a method to \emph{minimize} the \emph{total delay} required
for the base station to satisfy requests from all clients.  
\end{abstract}

\section{Introduction}
\label{sec:intro}

We study the issue of delay minimization of the so-called Direct Multicast with Side Information (DMSI) problem. In an instance of this problem, 
a base station 
wishes to transmit a set of $n$ original packets to a group of $k$ clients.
Each of the clients already has in its cache a subset of the original
packets (referred to as \emph{side information}), and requests for all the 
remaining packets. 
The base station communicates \emph{directly} with the clients by broadcasting 
information to them. 
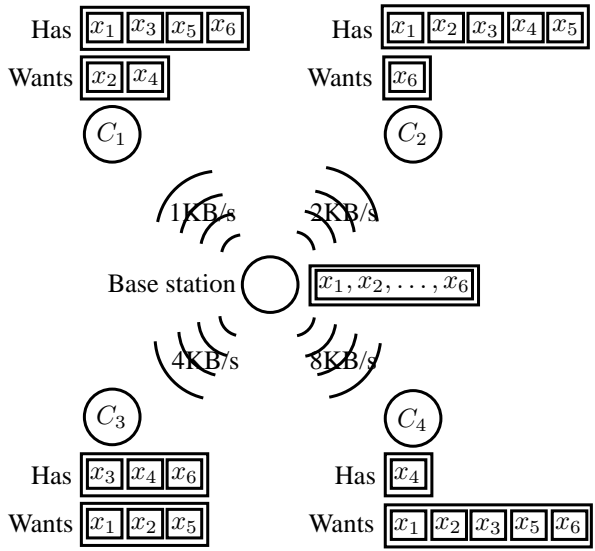
\begin{figure}[H]
\centering
\scalebox{1} 
{
\begin{pspicture}(0,-3.6200001)(7.9414062,3.6200001)
\pscircle[linewidth=0.04,dimen=outer](3.4945304,-0.11){0.39}
\pscircle[linewidth=0.04,dimen=outer](1.3945312,-1.8700001){0.39}
\usefont{T1}{ptm}{m}{n}
\rput(1.3859372,-1.8700001){$C_3$}
\pscircle[linewidth=0.04,dimen=outer](5.3945312,-1.8900001){0.39}
\usefont{T1}{ptm}{m}{n}
\rput(5.385937,-1.9100001){$C_4$}
\usefont{T1}{ptm}{m}{n}
\rput(2.189531,-0.09000021){Base station}
\pscircle[linewidth=0.04,dimen=outer](1.3945312,1.89){0.39}
\usefont{T1}{ptm}{m}{n}
\rput(1.3887496,1.89){$C_1$}
\psframe[linewidth=0.04,dimen=outer](3.2245312,3.6000001)(0.9645312,3.0)
\psframe[linewidth=0.04,dimen=outer](1.5445312,3.5199997)(1.0445312,3.1000001)
\usefont{T1}{ptm}{m}{n}
\rput(1.2767187,3.2900002){$x_1$}
\psframe[linewidth=0.04,dimen=outer](2.0845313,3.5199997)(1.5845312,3.1000001)
\usefont{T1}{ptm}{m}{n}
\rput(1.8209374,3.31){$x_3$}
\usefont{T1}{ptm}{m}{n}
\rput(0.5859375,3.2900002){Has}
\psframe[linewidth=0.04,dimen=outer](2.6045313,3.5199997)(2.1045313,3.1000001)
\usefont{T1}{ptm}{m}{n}
\rput(2.3567188,3.2900002){$x_5$}
\psframe[linewidth=0.04,dimen=outer](3.1445312,3.5199997)(2.6445312,3.1000001)
\usefont{T1}{ptm}{m}{n}
\rput(2.8809373,3.31){$x_6$}
\psframe[linewidth=0.04,dimen=outer](2.1645312,2.94)(0.9645312,2.3400002)
\psframe[linewidth=0.04,dimen=outer](1.5445312,2.8600001)(1.0445312,2.44)
\usefont{T1}{ptm}{m}{n}
\rput(1.2967187,2.63){$x_2$}
\psframe[linewidth=0.04,dimen=outer](2.0845313,2.8600001)(1.5845312,2.44)
\usefont{T1}{ptm}{m}{n}
\rput(1.8479687,2.65){$x_4$}
\usefont{T1}{ptm}{m}{n}
\rput(0.44671875,2.65){Wants}
\pscircle[linewidth=0.04,dimen=outer](5.3945312,1.89){0.39}
\usefont{T1}{ptm}{m}{n}
\rput(5.3887496,1.89){$C_2$}
\psframe[linewidth=0.04,dimen=outer](7.744531,3.6200001)(4.9645314,3.02)
\psframe[linewidth=0.04,dimen=outer](5.5445313,3.5199997)(5.0445313,3.1000001)
\usefont{T1}{ptm}{m}{n}
\rput(5.2767186,3.2900002){$x_1$}
\psframe[linewidth=0.04,dimen=outer](6.0845313,3.5199997)(5.5845313,3.1000001)
\usefont{T1}{ptm}{m}{n}
\rput(5.8209376,3.31){$x_2$}
\usefont{T1}{ptm}{m}{n}
\rput(4.5859375,3.2900002){Has}
\psframe[linewidth=0.04,dimen=outer](6.6045313,3.5199997)(6.1045313,3.1000001)
\usefont{T1}{ptm}{m}{n}
\rput(6.3567185,3.2900002){$x_3$}
\psframe[linewidth=0.04,dimen=outer](7.1445312,3.5199997)(6.6445312,3.1000001)
\usefont{T1}{ptm}{m}{n}
\rput(6.8809376,3.31){$x_4$}
\psframe[linewidth=0.04,dimen=outer](5.6445312,2.94)(4.9845314,2.3400002)
\psframe[linewidth=0.04,dimen=outer](5.5645313,2.8600001)(5.0645313,2.44)
\usefont{T1}{ptm}{m}{n}
\rput(5.2967186,2.63){$x_6$}
\usefont{T1}{ptm}{m}{n}
\rput(4.4267187,2.65){Wants}
\psframe[linewidth=0.04,dimen=outer](7.664531,3.5199997)(7.164531,3.1000001)
\usefont{T1}{ptm}{m}{n}
\rput(7.4209375,3.31){$x_5$}
\psframe[linewidth=0.04,dimen=outer](5.664531,-2.34)(5.0045314,-2.94)
\psframe[linewidth=0.04,dimen=outer](5.5845313,-2.42)(5.0845313,-2.84)
\usefont{T1}{ptm}{m}{n}
\rput(5.3167186,-2.6499999){$x_4$}
\usefont{T1}{ptm}{m}{n}
\rput(4.6259375,-2.6499999){Has}
\usefont{T1}{ptm}{m}{n}
\rput(4.4467187,-3.29){Wants}
\psframe[linewidth=0.04,dimen=outer](7.784531,-3.02)(5.0045314,-3.6200001)
\psframe[linewidth=0.04,dimen=outer](5.5845313,-3.1000001)(5.0845313,-3.52)
\usefont{T1}{ptm}{m}{n}
\rput(5.3167186,-3.33){$x_1$}
\psframe[linewidth=0.04,dimen=outer](6.1245313,-3.1000001)(5.6245313,-3.52)
\usefont{T1}{ptm}{m}{n}
\rput(5.8609376,-3.31){$x_2$}
\psframe[linewidth=0.04,dimen=outer](6.6445312,-3.1000001)(6.1445312,-3.52)
\usefont{T1}{ptm}{m}{n}
\rput(6.396719,-3.33){$x_3$}
\psframe[linewidth=0.04,dimen=outer](7.184531,-3.1000001)(6.684531,-3.52)
\usefont{T1}{ptm}{m}{n}
\rput(6.9209375,-3.31){$x_5$}
\psframe[linewidth=0.04,dimen=outer](2.6845312,-2.34)(0.9645312,-2.94)
\psframe[linewidth=0.04,dimen=outer](1.5445312,-2.42)(1.0445312,-2.84)
\usefont{T1}{ptm}{m}{n}
\rput(1.2767187,-2.6499999){$x_3$}
\psframe[linewidth=0.04,dimen=outer](2.0845313,-2.42)(1.5845312,-2.84)
\usefont{T1}{ptm}{m}{n}
\rput(1.8209374,-2.6499999){$x_4$}
\usefont{T1}{ptm}{m}{n}
\rput(0.5859375,-2.6499999){Has}
\psframe[linewidth=0.04,dimen=outer](2.6045313,-2.42)(2.1045313,-2.84)
\usefont{T1}{ptm}{m}{n}
\rput(2.3567188,-2.6499999){$x_6$}
\usefont{T1}{ptm}{m}{n}
\rput(0.44671875,-3.29){Wants}
\psframe[linewidth=0.04,dimen=outer](2.6845312,-3.0)(0.9645312,-3.6000001)
\psframe[linewidth=0.04,dimen=outer](1.5445312,-3.0800002)(1.0445312,-3.5)
\usefont{T1}{ptm}{m}{n}
\rput(1.2767187,-3.31){$x_1$}
\psframe[linewidth=0.04,dimen=outer](2.0845313,-3.0800002)(1.5845312,-3.5)
\usefont{T1}{ptm}{m}{n}
\rput(1.8209374,-3.31){$x_2$}
\psframe[linewidth=0.04,dimen=outer](2.6045313,-3.0800002)(2.1045313,-3.5)
\usefont{T1}{ptm}{m}{n}
\rput(2.3567188,-3.31){$x_5$}
\psframe[linewidth=0.04,dimen=outer](6.284531,0.16000006)(4.0045314,-0.39999995)
\psframe[linewidth=0.04,dimen=outer](6.204531,0.09999985)(4.0845313,-0.31999996)
\usefont{T1}{ptm}{m}{n}
\rput(5.146719,-0.10999985){$x_1, x_2, \ldots, x_6$}
\usefont{T1}{ptm}{m}{n}
\rput(2.5971875,0.85){1KB/s}
\rput{76.81656}(1.9001064,-3.4188836){\psarc[linewidth=0.04](3.1061919,-0.5111286){0.28}{106.92751}{180.0}}
\rput{76.81656}(1.7793007,-3.3973324){\psarc[linewidth=0.04](3.0321977,-0.57653993){0.49}{106.92751}{180.0}}
\rput{76.81656}(1.5656121,-3.372439){\psarc[linewidth=0.04](2.9096541,-0.6988571){0.63}{106.92751}{180.0}}
\rput{76.81656}(1.3685316,-3.3461447){\psarc[linewidth=0.04](2.7945313,-0.8099999){0.83}{106.92751}{180.0}}
\rput{173.1648}(7.517945,-1.4755596){\psarc[linewidth=0.04](3.803032,-0.51329684){0.28}{106.92751}{180.0}}
\rput{173.1648}(7.6559167,-1.6164148){\psarc[linewidth=0.04](3.8762238,-0.5796046){0.49}{106.92751}{180.0}}
\rput{173.1648}(7.912305,-1.8482611){\psarc[linewidth=0.04](4.0113406,-0.6878721){0.63}{106.92751}{180.0}}
\rput{173.1648}(8.145656,-2.066452){\psarc[linewidth=0.04](4.1345315,-0.7899999){0.83}{106.92751}{180.0}}
\rput{260.05392}(4.1473436,4.111372){\psarc[linewidth=0.04](3.800246,0.31400552){0.28}{106.92751}{180.0}}
\rput{260.05392}(4.1612062,4.2619867){\psarc[linewidth=0.04](3.870428,0.38349104){0.49}{106.92751}{180.0}}
\rput{260.05392}(4.169483,4.527023){\psarc[linewidth=0.04](3.9858687,0.5125334){0.63}{106.92751}{180.0}}
\rput{260.05392}(4.181213,4.771808){\psarc[linewidth=0.04](4.094531,0.6300001){0.83}{106.92751}{180.0}}
\rput{349.34924}(0.0036259675,0.5826009){\psarc[linewidth=0.04](3.1268826,0.27185076){0.28}{106.92751}{180.0}}
\rput{349.34924}(-0.010684349,0.5711427){\psarc[linewidth=0.04](3.0582654,0.34288216){0.49}{106.92751}{180.0}}
\rput{349.34924}(-0.034510553,0.54957294){\psarc[linewidth=0.04](2.9306526,0.459901){0.63}{106.92751}{180.0}}
\rput{349.34924}(-0.05685986,0.53000796){\psarc[linewidth=0.04](2.8145313,0.5700001){0.83}{106.92751}{180.0}}
\usefont{T1}{ptm}{m}{n}
\rput(4.474375,0.85){2KB/s}
\usefont{T1}{ptm}{m}{n}
\rput(4.467344,-1.1099999){8KB/s}
\usefont{T1}{ptm}{m}{n}
\rput(2.6354687,-1.1099999){4KB/s}
\psframe[linewidth=0.04,dimen=outer](7.724531,-3.1000001)(7.224531,-3.52)
\usefont{T1}{ptm}{m}{n}
\rput(7.4609375,-3.31){$x_6$}
\end{pspicture} 
}
\caption{An example of Direct Multicast with Side Information}
\label{fig:dmsi_example}
\end{figure} 
\vspace{-5pt}
Such a scenario is usually observed in opportunistic wireless networks
\cite{Katti2006,Katti2008}, where wireless nodes often opportunistically overhear
packets that are not designated to them. These overheard packets become
the side information for the nodes. 
This problem also arises in communication schemes where a server has
to broadcast a set of packets to a group of clients. 
Limited storage capacity, bad reception, or signal degradation might
lead to packet loss at the clients. Using a slow feedback channel,
the clients inform the server about their missing packets, and request
for retransmissions \cite{BirkKol2006}.      

In our model, each packet that is transmitted from the base station, 
referred to as a \emph{broadcast} packet, is a linear combination of the \emph{original} packets.
Assume that bandwidths \emph{vary} between the base station and different
clients, and that each broadcast packet is designated for (in other words, assigned to) 
\emph{a subgroup} of clients.  
The \emph{delay} of a broadcast packet is defined to be the amount of time that 
a client (to which the packet is assigned) with a minimum bandwidth 
can receive the packet successfully. 
\emph{Our main contribution} is to provide a method to \emph{minimize} 
the \emph{total delay} required
for the base station to satisfy requests from all clients.  
We design an optimal packet assignment so as to achieve the minimum total delay.
Moreover, the multicast scheme with optimal total delay can be found in 
polynomial time in $n$ and $k$. 

\nin {\bf A motivational example.}
Suppose that there are four clients $C_1$, $C_2$, $C_3$, $C_4$, 
which miss $2$, $1$, $3$, $5$ original packets, respectively, as given in
Fig~\ref{fig:dmsi_example}.
By a well-known result in network coding (see 
Section~\ref{sec:feasibility} for more details), provided that 
\begin{itemize}
	\item the base station broadcasts at least $5$ packets to the clients, and
	\item the number of broadcast packets designated for each client is as many as 
	the number of its missing packets,  
\end{itemize}
then there is a coding scheme for the base station to satisfy demands 
from all clients simultaneously. Assume that each broadcast packet is of size $8$KB
and that the bandwidth and the packet delay from each client are given
the table in Fig.~\ref{fig:delay-table}. 
Note that the delay is obtained by dividing the packet
size by the bandwidth between the corresponding client and the base station. 
\vspace{-5pt}
\begin{figure}[ht]
\centering
		\begin{tabular}{|c|c|c|c|c|}
		\hline 
		& $C_1$ & $C_2$ & $C_3$ & $C_4$\\
		\hline 
		Bandwidth (KB/sec)& $1$ & $2$ & $4$ & $8$\\
		\hline 		
		Delay (sec)& $8$ & $4$ & $2$ & $1$\\
		\hline
		\end{tabular}
		\caption{Bandwidths and delays for clients}
   \label{fig:delay-table}
\vspace{-5pt}
\end{figure}
Suppose that the base station uses five broadcast packets $p_1$,$\ldots$, $p_5$.
Consider the Packet Assignment A, given in Fig.~\ref{fig:paa}, and
the Packet Assignment B, given in Fig.~\ref{fig:pab}.
The total delay of the Packet Assignment B ($20$ seconds) is $4$ seconds less than the
total delay of the Packet Assignment A ($24$ seconds). 
In fact, in Section~\ref{sec:optimality}, we can see that Packet Assignment B is actually
optimal in terms of the total delay for this scenario. 
The intuition is that the total delay gets smaller if fewer broadcast packets are assigned to 
more clients with large delays. This is proved later to be true. 
\vspace{-5pt}
\begin{figure}[ht]
	\centering
		\begin{tabular}{c|c|c|c|c|c|}
		& $C_1$ & $C_2$ & $C_3$ & $C_4$ & Packet delay (sec)\\
		\hline
		$p_1$ & 1 & & 1 & 1 & 8\\ 
		\hline
		$p_2$ &  & 1 &  & 1 & 4\\
		\hline
		$p_3$ & 1 & & & 1 & 8\\
		\hline
		$p_4$ & & & 1 & 1 & 2\\
		\hline
		$p_5$ & & & 1 & 1 & 2\\
		\hline
		Total delay & & & & & {\bf 24}\\
		\hline
		\end{tabular}
		\caption{Packet Assignment A. A $1$-entry means the 
		broadcast packet in that row is assigned to the client in that column. 
		The delay of a broadcast packet is the maximum delay from all clients
	 to which the packet is assigned.}
		\label{fig:paa}
\end{figure}
\vspace{-20pt}
\begin{figure}[ht]
	\centering
		\begin{tabular}{c|c|c|c|c|c|}
		& $C_1$ & $C_2$ & $C_3$ & $C_4$ & Packet delay (sec)\\
		\hline
		$p_1$ & 1 & 1 & 1 & 1 & 8\\ 
		\hline
		$p_2$ & 1  &  & 1  & 1 & 8\\
		\hline
		$p_3$ &  & & 1 & 1 & 2\\
		\hline
		$p_4$ & & &  & 1 & 1\\
		\hline
		$p_5$ & & &  & 1 & 1\\
		\hline
		Total delay & & & & & {\bf 20}\\
		\hline
		\end{tabular}
		\caption{Packet Assignment B}
		\label{fig:pab}
\end{figure}
\vspace{-10pt}

\nin {\bf Related work.}
The DMSI problem is 
a special case of
the Multicast with Side Information (MSI) problem~\cite{BakshiEffros2008}. 
In an MSI instance, there is a network between the base station and the clients. 
Our problem considers the
scenario where the only communication links are those between the base
station and the clients. 
However, the issue of delay minimization is not investigated in~\cite{BakshiEffros2008}.  

Lun {\et}~\cite{Lun-etal2006} study the problem of cost minimization 
for a general multicast network. In their setting, each vector of rates $\bz$ 
at which packets are injected into edges of the network corresponds to
a cost $f(\bz)$. The goal is to find $\bz$ that minimizes $f(\bz)$. 
The main difference between our result and the result in~\cite{Lun-etal2006}
is the following.
The authors in~\cite{Lun-etal2006} investigate asymptotic solutions
with infinite block length codes; in other words, they consider divisible packets 
with infinitely many subpackets (the non-integral setting). 
In this work, we are only interested in network codes of block length one; 
in other words, we only consider indivisible packets (the integral setting).
From a practical point of view, solutions to the integral setting are often 
preferred due to its simplicity in implementation, lower complexity in 
computation, and smaller buffer required at clients.  
In general, the integral setting might be harder to tackle than the non-integral
setting (for instance linear programming can be solved in polynomial time, 
whereas integer linear programming is NP-hard). 
However, in our case, because of the special objective function
(the total delay), the optimal solution for the integral setting can be found
in polynomial time.   

\nin{\bf Organization.} We formulate our problem rigorously in 
Section~\ref{sec:probdef}. A necessary and sufficient condition 
for the feasibility of a multicast scheme is provided in Section~\ref{sec:feasibility}
(Lemma~\ref{lem:feasibility}). In Section~\ref{sec:optimality}, 
we construct a feasible multicast scheme and prove that it has minimum 
total delay (Lemma~\ref{lem:optimality}, Theorem~\ref{thm:main}). 

\section{Problem Definition}
\label{sec:probdef}

A Direct Multicast with Side Information (DMSI) instance is described 
as follows. 
A base station $S$ has a set of $n$ \emph{original} packets
$X = \{x_1, \ldots, x_n\}$, where $x_i \in \fq$, $i \in [n]$. 
There are $k$ clients $C_1, \ldots, C_k$. 
For each $j \in [k]$, the client $C_j$ possesses a subset of original packets $H_j
\seq X$ as side information, and demands all missing packets in $X \setminus H_j$. 
We abbreviate such a DMSI instance by $\M = (n,\{H_j\}^k_1)$. 

A \emph{multicast scheme} for the instance $\M=\nhj$
is a $2$-tuple $\pa$ where
\begin{itemize}
	\item $P = \{p_1,\ldots, p_m\}$ is a set of \emph{broadcast packets}, i.e. linear combinations
	of the original packets, that the	base station broadcasts to the clients,  
	\item $\bA = (a_{i,j})$ is an $m \times k$ binary matrix, where $a_{i,j} = 1$ 
	if and only if the broadcast packet $p_i$ is assigned to the client $C_j$, for $i \in [m]$, $j \in [k]$.  
\end{itemize}
We refer to $\bA$ as the (packet) \emph{assignment matrix}. The assignment matrix determines which clients a broadcast packet
is assigned to. 
A multicast scheme is \emph{feasible} if upon receiving all designated broadcast packets, 
each client can retrieve all missing original packets. 

We assume that the client $C_j$ $(j \in [k])$ requires $d_j$ seconds to receipt
a broadcast packet (assigned to it) successfully. We refer to $d_j$ as the \emph{delay} 
from $C_j$ $(j \in [k])$. 
Furthermore, suppose that after broadcasting a packet $p_i \in P$, the base station can start sending
another packet only when all clients that $p_i$ is designated for already receive $p_i$
successfully. We define the \emph{delay} of the packet $p_i$ according to the 
assignment matrix $\bA$ by
\vspace{-8pt}
\begin{equation}
\label{eq:packet_delay}
\dla(p_i) = \max \{d_j: \ a_{i,j} = 1\}.  
\vspace{-3pt}
\end{equation} 
Note that the base station must transmit $p_i$ at the minimum rate among all designated clients so that the client with the smallest bandwidth can manage to decode the packet.
Therefore, the largest delay among the designated clients is the bottleneck and 
dominates the delay for that broadcast packet transmission. 
Therefore, $\dla(p_i)$ is the amount of time required for the broadcast packet $p_i$ to be successfully
received by all designated clients.
We define the \emph{total delay} of a multicast scheme $\pa$ by
\vspace{-8pt}
\begin{equation} 
\label{eq:total_delay}
\dt\pa = \sum_{i = 1}^m \dla(p_i).  
\vspace{-3pt}
\end{equation}  
Notice that the total delay can be determined solely from the assignment matrix. 
Therefore, sometimes we use $\dt(\bA)$ instead of $\dt\pa$.
\emph{Our goal} is to find a feasible multicast scheme with 
\emph{minimum total delay}, for a given DMSI instance. 

As an illustrative example, we consider the DMSI instance 
as described in Fig.~\ref{fig:dmsi_example}. 
In this example, $n = 6$, $k = 4$, and the side information at the clients are 
given below.
\vspace{-5pt}
\[
\begin{split}
H_1 &= \{x_1,x_3,x_5,x_6\}, \ H_2 = \{x_1,x_2,x_3,x_4,x_5\},\\
H_3 &= \{x_3,x_4,x_6\}, \ H_4 = \{x_4\}. 
\end{split}
\]
The Packet Assignment~A and~B in Fig.~\ref{fig:paa} and~\ref{fig:pab} can be incorporated
into multicast schemes $\pa$ and $\pas$, respectively, where the assignment matrices are 
given in Fig.~\ref{fig:matrices}. 
In Section~\ref{sec:optimality}, we show how to determine the packets in $P$ and $P^*$ 
so that $\pa$ and $\pas$ are feasible multicast schemes for $\M$.
\vspace{-10pt}
\begin{figure}[ht]
\centering
\subfloat[]{
$ 
\bA =
\begin{pmatrix}
1 & 0 & 1 & 1\\
0 & 1 & 0 & 1\\
1 & 0 & 0 & 1\\
0 & 0 & 1 & 1\\
0 & 0 & 1 & 1
\end{pmatrix},
$
\label{mat:A}
}
\subfloat[]{
$ 
\bas =
\begin{pmatrix}
1 & 1 & 1 & 1\\
1 & 0 & 1 & 1\\
0 & 0 & 1 & 1\\
0 & 0 & 0 & 1\\
0 & 0 & 0 & 1
\end{pmatrix}.
$
\label{mat:Astar}
}
\caption{The assignment matrices $\bA$ and $\bas$}
\label{fig:matrices}
\vspace{-8pt}
\end{figure}
Regarding the packet delay, let us examine the third broadcast packet $p_3$,
which is designated for $C_3$ and $C_4$, according to $\bas$.  
The delay from these two clients are $2=8/4$ seconds and $1=8/8$ seconds, 
respectively. Therefore, 
\vspace{-3pt}
\[
\dlas(p_3) = \max\{2,1\} = 2.
\vspace{-3pt}
\] 
The total delay of the matrix $\bas$ is calculated as follows. 
\vspace{-3pt}
\[
\begin{split} 
\dt(\bas) &= \sum_{i = 1}^5 \dlas(p_i)\\
&= \max\{8, 4, 2, 1\} + \max\{8, 2, 1\} + \max\{2, 1\}\\
&\quad + \max\{1\} + \max\{1\}\\
&= 8 + 8 + 2 + 1 + 1 = 20.
\end{split}
\]
\section{Feasibility of a Multicast Scheme via Network Coding}
\label{sec:feasibility}

In this section, we establish a necessary and sufficient condition for the 
feasibility of a multicast scheme for DMSI via network coding. 

Hereafter, let $\M = \nhj$ be a DMSI instance. 
For an $m \times k$ binary matrix $\bA$, we define
the network $\nma$ as follows. The set of nodes of $\nma$ consists of
\begin{itemize}
	\item one source node $s$, which possesses all original packets $x_1, \ldots, x_n$,
	\item $n$ ``original packet" nodes $s_1, \ldots, s_n$, each corresponds
		to an original packet,
	\item $m$ intermediate nodes $u_1, \ldots, u_m$,
	\item $m$ ``broadcast packet" nodes $v_1, \ldots, v_m$,	
	\item $k$ sinks $t_1, t_2, \ldots, t_k$, each corresponds to a client and demands
	all original packets. 
\end{itemize}
The set of (directed) edges of $\nma$ consists of
\begin{itemize}
	\item $(s,s_i)$ with capacity one for all $i \in [n]$,
	\item $(s_i,t_j)$ with capacity infinity if and only if $x_i \in H_j$,
	\item $(s_i,u_h)$ with capacity infinity for all $i \in [n]$, $h \in [m]$,
  \item $(u_h, v_h)$ with capacity one for every $h \in [m]$,
	\item $(v_h,t_j)$ with capacity one if and only if $a_{h,j} = 1$.
\end{itemize}
As an illustrative example, the network $\nmas$, where $\M$ is given in Fig.~\ref{fig:dmsi_example}
and $\bas$ is given in Fig.~\ref{mat:Astar}, is depicted in Fig.~\ref{fig:network_example}
and Fig.~\ref{fig:si_edges}. 
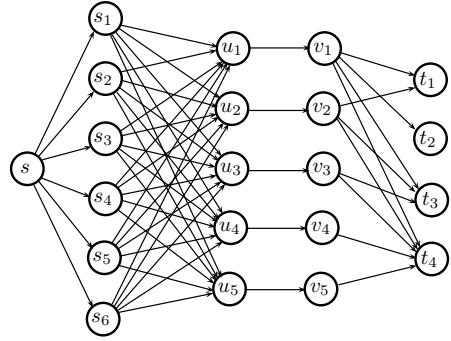
\begin{figure}[H]
\centering
\scalebox{0.8} 
{
\begin{pspicture}(0,-2.79)(7.443281,2.79)
\pscircle[linewidth=0.04,dimen=outer](0.29,0.0){0.29}
\usefont{T1}{ptm}{m}{n}
\rput(0.2528125,0.0){$s$}
\pscircle[linewidth=0.04,dimen=outer](1.59,2.5){0.29}
\usefont{T1}{ptm}{m}{n}
\rput(1.5528125,2.5){$s_1$}
\pscircle[linewidth=0.04,dimen=outer](1.59,1.5){0.29}
\usefont{T1}{ptm}{m}{n}
\rput(1.5528125,1.5){$s_2$}
\pscircle[linewidth=0.04,dimen=outer](1.59,0.5){0.29}
\usefont{T1}{ptm}{m}{n}
\rput(1.5528125,0.5){$s_3$}
\pscircle[linewidth=0.04,dimen=outer](1.59,-0.5){0.29}
\usefont{T1}{ptm}{m}{n}
\rput(1.5528125,-0.5){$s_4$}
\pscircle[linewidth=0.04,dimen=outer](1.57,-1.48){0.29}
\usefont{T1}{ptm}{m}{n}
\rput(1.5328126,-1.48){$s_5$}
\pscircle[linewidth=0.04,dimen=outer](1.55,-2.5){0.29}
\usefont{T1}{ptm}{m}{n}
\rput(1.5128125,-2.5){$s_6$}
\pscircle[linewidth=0.04,dimen=outer](3.71,2.0){0.29}
\usefont{T1}{ptm}{m}{n}
\rput(3.6928124,2.0){$u_1$}
\pscircle[linewidth=0.04,dimen=outer](3.69,1.0){0.29}
\usefont{T1}{ptm}{m}{n}
\rput(3.6728125,1.0){$u_2$}
\pscircle[linewidth=0.04,dimen=outer](3.69,0.0){0.29}
\usefont{T1}{ptm}{m}{n}
\rput(3.6728125,0.0){$u_3$}
\pscircle[linewidth=0.04,dimen=outer](3.67,-0.98){0.29}
\usefont{T1}{ptm}{m}{n}
\rput(3.6528125,-0.98){$u_4$}
\pscircle[linewidth=0.04,dimen=outer](3.65,-2.0){0.29}
\usefont{T1}{ptm}{m}{n}
\rput(3.6328125,-2.0){$u_5$}
\pscircle[linewidth=0.04,dimen=outer](5.23,2.0){0.29}
\usefont{T1}{ptm}{m}{n}
\rput(5.2128124,2.0){$v_1$}
\pscircle[linewidth=0.04,dimen=outer](5.21,1.0){0.29}
\usefont{T1}{ptm}{m}{n}
\rput(5.1928124,1.0){$v_2$}
\pscircle[linewidth=0.04,dimen=outer](5.21,0.0){0.29}
\usefont{T1}{ptm}{m}{n}
\rput(5.1928124,0.0){$v_3$}
\pscircle[linewidth=0.04,dimen=outer](5.19,-0.98){0.29}
\usefont{T1}{ptm}{m}{n}
\rput(5.1728125,-0.98){$v_4$}
\pscircle[linewidth=0.04,dimen=outer](5.17,-2.0){0.29}
\usefont{T1}{ptm}{m}{n}
\rput(5.1528125,-2.0){$v_5$}
\pscircle[linewidth=0.04,dimen=outer](6.99,1.48){0.29}
\usefont{T1}{ptm}{m}{n}
\rput(6.9528127,1.48){$t_1$}
\pscircle[linewidth=0.04,dimen=outer](6.99,0.5){0.29}
\usefont{T1}{ptm}{m}{n}
\rput(6.9528127,0.5){$t_2$}
\pscircle[linewidth=0.04,dimen=outer](7.01,-0.52){0.29}
\usefont{T1}{ptm}{m}{n}
\rput(6.9928126,-0.52){$t_3$}
\pscircle[linewidth=0.04,dimen=outer](7.01,-1.5){0.29}
\usefont{T1}{ptm}{m}{n}
\rput(6.9728127,-1.5){$t_4$}
\psline[linewidth=0.02cm,arrowsize=0.05291667cm 2.0,arrowlength=1.4,arrowinset=0.4]{->}(0.42,0.25)(1.4,2.29)
\psline[linewidth=0.02cm,arrowsize=0.05291667cm 2.0,arrowlength=1.4,arrowinset=0.4]{->}(0.48,0.21)(1.4,1.29)
\psline[linewidth=0.02cm,arrowsize=0.05291667cm 2.0,arrowlength=1.4,arrowinset=0.4]{->}(0.54,0.15)(1.36,0.37)
\psline[linewidth=0.02cm,arrowsize=0.05291667cm 2.0,arrowlength=1.4,arrowinset=0.4]{->}(0.52,-0.15)(1.34,-0.47)
\psline[linewidth=0.02cm,arrowsize=0.05291667cm 2.0,arrowlength=1.4,arrowinset=0.4]{->}(0.48,-0.21)(1.38,-1.33)
\psline[linewidth=0.02cm,arrowsize=0.05291667cm 2.0,arrowlength=1.4,arrowinset=0.4]{->}(0.4,-0.27)(1.38,-2.31)
\psline[linewidth=0.02cm,arrowsize=0.05291667cm 2.0,arrowlength=1.4,arrowinset=0.4]{->}(1.84,2.39)(3.46,2.03)
\psline[linewidth=0.02cm,arrowsize=0.05291667cm 2.0,arrowlength=1.4,arrowinset=0.4]{->}(1.82,2.33)(3.5,1.17)
\psline[linewidth=0.02cm,arrowsize=0.05291667cm 2.0,arrowlength=1.4,arrowinset=0.4]{->}(1.76,2.27)(3.52,0.21)
\psline[linewidth=0.02cm,arrowsize=0.05291667cm 2.0,arrowlength=1.4,arrowinset=0.4]{->}(1.7,2.23)(3.54,-0.75)
\psline[linewidth=0.02cm,arrowsize=0.05291667cm 2.0,arrowlength=1.4,arrowinset=0.4]{->}(1.66,2.25)(3.52,-1.81)
\psline[linewidth=0.02cm,arrowsize=0.05291667cm 2.0,arrowlength=1.4,arrowinset=0.4]{->}(1.84,1.61)(3.44,1.97)
\psline[linewidth=0.02cm,arrowsize=0.05291667cm 2.0,arrowlength=1.4,arrowinset=0.4]{->}(1.86,1.47)(3.46,1.03)
\psline[linewidth=0.02cm,arrowsize=0.05291667cm 2.0,arrowlength=1.4,arrowinset=0.4]{->}(1.82,1.35)(3.48,0.13)
\psline[linewidth=0.02cm,arrowsize=0.05291667cm 2.0,arrowlength=1.4,arrowinset=0.4]{->}(1.78,1.27)(3.48,-0.85)
\psline[linewidth=0.02cm,arrowsize=0.05291667cm 2.0,arrowlength=1.4,arrowinset=0.4]{->}(1.72,1.25)(3.46,-1.83)
\psline[linewidth=0.02cm,arrowsize=0.05291667cm 2.0,arrowlength=1.4,arrowinset=0.4]{->}(1.84,0.63)(3.54,1.79)
\psline[linewidth=0.02cm,arrowsize=0.05291667cm 2.0,arrowlength=1.4,arrowinset=0.4]{->}(1.86,0.55)(3.44,0.91)
\psline[linewidth=0.02cm,arrowsize=0.05291667cm 2.0,arrowlength=1.4,arrowinset=0.4]{->}(1.86,0.39)(3.42,0.03)
\psline[linewidth=0.02cm,arrowsize=0.05291667cm 2.0,arrowlength=1.4,arrowinset=0.4]{->}(1.82,0.33)(3.46,-0.89)
\psline[linewidth=0.02cm,arrowsize=0.05291667cm 2.0,arrowlength=1.4,arrowinset=0.4]{->}(1.76,0.31)(3.44,-1.87)
\psline[linewidth=0.02cm,arrowsize=0.05291667cm 2.0,arrowlength=1.4,arrowinset=0.4]{->}(1.76,-0.27)(3.56,1.77)
\psline[linewidth=0.02cm,arrowsize=0.05291667cm 2.0,arrowlength=1.4,arrowinset=0.4]{->}(1.8,-0.31)(3.48,0.87)
\psline[linewidth=0.02cm,arrowsize=0.05291667cm 2.0,arrowlength=1.4,arrowinset=0.4]{->}(1.86,-0.43)(3.44,-0.11)
\psline[linewidth=0.02cm,arrowsize=0.05291667cm 2.0,arrowlength=1.4,arrowinset=0.4]{->}(1.88,-0.53)(3.42,-0.97)
\psline[linewidth=0.02cm,arrowsize=0.05291667cm 2.0,arrowlength=1.4,arrowinset=0.4]{->}(1.78,-0.67)(3.4,-1.91)
\psline[linewidth=0.02cm,arrowsize=0.05291667cm 2.0,arrowlength=1.4,arrowinset=0.4]{->}(1.7,-1.23)(3.62,1.77)
\psline[linewidth=0.02cm,arrowsize=0.05291667cm 2.0,arrowlength=1.4,arrowinset=0.4]{->}(1.76,-1.27)(3.56,0.77)
\psline[linewidth=0.02cm,arrowsize=0.05291667cm 2.0,arrowlength=1.4,arrowinset=0.4]{->}(1.8,-1.33)(3.52,-0.21)
\psline[linewidth=0.02cm,arrowsize=0.05291667cm 2.0,arrowlength=1.4,arrowinset=0.4]{->}(1.84,-1.45)(3.44,-1.11)
\psline[linewidth=0.02cm,arrowsize=0.05291667cm 2.0,arrowlength=1.4,arrowinset=0.4]{->}(1.82,-1.61)(3.4,-2.03)
\psline[linewidth=0.02cm,arrowsize=0.05291667cm 2.0,arrowlength=1.4,arrowinset=0.4]{->}(1.7,-2.25)(3.66,1.73)
\psline[linewidth=0.02cm,arrowsize=0.05291667cm 2.0,arrowlength=1.4,arrowinset=0.4]{->}(1.72,-2.31)(3.64,0.75)
\psline[linewidth=0.02cm,arrowsize=0.05291667cm 2.0,arrowlength=1.4,arrowinset=0.4]{->}(1.76,-2.33)(3.58,-0.25)
\psline[linewidth=0.02cm,arrowsize=0.05291667cm 2.0,arrowlength=1.4,arrowinset=0.4]{->}(1.8,-2.35)(3.54,-1.21)
\psline[linewidth=0.02cm,arrowsize=0.05291667cm 2.0,arrowlength=1.4,arrowinset=0.4]{->}(3.9,-2.01)(4.92,-2.01)
\psline[linewidth=0.02cm,arrowsize=0.05291667cm 2.0,arrowlength=1.4,arrowinset=0.4]{->}(3.92,-0.99)(4.94,-0.99)
\psline[linewidth=0.02cm,arrowsize=0.05291667cm 2.0,arrowlength=1.4,arrowinset=0.4]{->}(3.96,-0.03)(4.98,-0.03)
\psline[linewidth=0.02cm,arrowsize=0.05291667cm 2.0,arrowlength=1.4,arrowinset=0.4]{->}(3.96,0.97)(4.98,0.97)
\psline[linewidth=0.02cm,arrowsize=0.05291667cm 2.0,arrowlength=1.4,arrowinset=0.4]{->}(3.98,2.01)(5.0,2.01)
\psline[linewidth=0.02cm,arrowsize=0.05291667cm 2.0,arrowlength=1.4,arrowinset=0.4]{->}(5.5,1.93)(6.74,1.55)
\psline[linewidth=0.02cm,arrowsize=0.05291667cm 2.0,arrowlength=1.4,arrowinset=0.4]{->}(5.46,1.87)(6.74,0.65)
\psline[linewidth=0.02cm,arrowsize=0.05291667cm 2.0,arrowlength=1.4,arrowinset=0.4]{->}(5.4,1.81)(6.8,-0.35)
\psline[linewidth=0.02cm,arrowsize=0.05291667cm 2.0,arrowlength=1.4,arrowinset=0.4]{->}(5.36,1.77)(6.82,-1.31)
\psline[linewidth=0.02cm,arrowsize=0.05291667cm 2.0,arrowlength=1.4,arrowinset=0.4]{->}(5.44,0.87)(6.74,-0.39)
\psline[linewidth=0.02cm,arrowsize=0.05291667cm 2.0,arrowlength=1.4,arrowinset=0.4]{->}(5.44,0.83)(6.78,-1.35)
\psline[linewidth=0.02cm,arrowsize=0.05291667cm 2.0,arrowlength=1.4,arrowinset=0.4]{->}(5.5,1.05)(6.74,1.35)
\psline[linewidth=0.02cm,arrowsize=0.05291667cm 2.0,arrowlength=1.4,arrowinset=0.4]{->}(5.46,-0.11)(6.76,-1.39)
\psline[linewidth=0.02cm,arrowsize=0.05291667cm 2.0,arrowlength=1.4,arrowinset=0.4]{->}(5.46,-0.07)(6.76,-0.57)
\psline[linewidth=0.02cm,arrowsize=0.05291667cm 2.0,arrowlength=1.4,arrowinset=0.4]{->}(5.42,-1.09)(6.76,-1.49)
\psline[linewidth=0.02cm,arrowsize=0.05291667cm 2.0,arrowlength=1.4,arrowinset=0.4]{->}(5.46,-1.93)(6.76,-1.61)
\psline[linewidth=0.02cm,arrowsize=0.05291667cm 2.0,arrowlength=1.4,arrowinset=0.4]{->}(1.8,-2.39)(3.38,-2.07)
\end{pspicture} 
}
\caption{The network $\nmas$ with $\M$ given in Fig.~\ref{fig:dmsi_example} and $\bas$ 
given in Fig.~\ref{mat:Astar}. The (side information) edges from $s_i$ to $t_j$ are depicted in 
a separate figure (Fig.~\ref{fig:si_edges}) for a clearer view}
\label{fig:network_example}
\vspace{-15pt}
\end{figure} 
\begin{figure}[H]
\centering
\scalebox{0.8} 
{
\begin{pspicture}(0,-2.79)(6.3228126,2.79)
\pscircle[linewidth=0.04,dimen=outer](0.4509375,2.5){0.29}
\usefont{T1}{ptm}{m}{n}
\rput(0.43234375,2.5){$s_1$}
\pscircle[linewidth=0.04,dimen=outer](0.4509375,1.5){0.29}
\usefont{T1}{ptm}{m}{n}
\rput(0.43234375,1.5){$s_2$}
\pscircle[linewidth=0.04,dimen=outer](0.4509375,0.5){0.29}
\usefont{T1}{ptm}{m}{n}
\rput(0.43234375,0.5){$s_3$}
\pscircle[linewidth=0.04,dimen=outer](0.4509375,-0.5){0.29}
\usefont{T1}{ptm}{m}{n}
\rput(0.43234375,-0.5){$s_4$}
\pscircle[linewidth=0.04,dimen=outer](0.4309375,-1.48){0.29}
\usefont{T1}{ptm}{m}{n}
\rput(0.41234374,-1.48){$s_5$}
\pscircle[linewidth=0.04,dimen=outer](0.4109375,-2.5){0.29}
\usefont{T1}{ptm}{m}{n}
\rput(0.39234376,-2.5){$s_6$}
\pscircle[linewidth=0.04,dimen=outer](5.8509374,1.48){0.29}
\usefont{T1}{ptm}{m}{n}
\rput(5.8323436,1.48){$t_1$}
\pscircle[linewidth=0.04,dimen=outer](5.8509374,0.5){0.29}
\usefont{T1}{ptm}{m}{n}
\rput(5.8323436,0.5){$t_2$}
\pscircle[linewidth=0.04,dimen=outer](5.8709373,-0.52){0.29}
\usefont{T1}{ptm}{m}{n}
\rput(5.8723435,-0.52){$t_3$}
\pscircle[linewidth=0.04,dimen=outer](5.8709373,-1.5){0.29}
\usefont{T1}{ptm}{m}{n}
\rput(5.8523436,-1.5){$t_4$}
\psline[linewidth=0.02cm,arrowsize=0.05291667cm 2.0,arrowlength=1.4,arrowinset=0.4]{->}(0.7209375,2.47)(5.6009374,1.49)
\psline[linewidth=0.02cm,arrowsize=0.05291667cm 2.0,arrowlength=1.4,arrowinset=0.4]{->}(0.7009375,0.51)(5.5809374,1.37)
\psline[linewidth=0.02cm,arrowsize=0.05291667cm 2.0,arrowlength=1.4,arrowinset=0.4]{->}(0.6809375,-1.41)(5.6809373,1.27)
\psline[linewidth=0.02cm,arrowsize=0.05291667cm 2.0,arrowlength=1.4,arrowinset=0.4]{->}(0.6809375,-2.45)(5.7209377,1.25)
\psline[linewidth=0.02cm,arrowsize=0.05291667cm 2.0,arrowlength=1.4,arrowinset=0.4]{->}(0.7009375,2.37)(5.6009374,0.61)
\psline[linewidth=0.02cm,arrowsize=0.05291667cm 2.0,arrowlength=1.4,arrowinset=0.4]{->}(0.7009375,1.43)(5.6209373,0.49)
\psline[linewidth=0.02cm,arrowsize=0.05291667cm 2.0,arrowlength=1.4,arrowinset=0.4]{->}(0.7009375,0.43)(5.6009374,0.41)
\psline[linewidth=0.02cm,arrowsize=0.05291667cm 2.0,arrowlength=1.4,arrowinset=0.4]{->}(0.7409375,-0.49)(5.6409373,0.37)
\psline[linewidth=0.02cm,arrowsize=0.05291667cm 2.0,arrowlength=1.4,arrowinset=0.4]{->}(0.7009375,-1.45)(5.6809373,0.33)
\psline[linewidth=0.02cm,arrowsize=0.05291667cm 2.0,arrowlength=1.4,arrowinset=0.4]{->}(0.7009375,0.37)(5.6409373,-0.47)
\psline[linewidth=0.02cm,arrowsize=0.05291667cm 2.0,arrowlength=1.4,arrowinset=0.4]{->}(0.7409375,-0.53)(5.6009374,-0.55)
\psline[linewidth=0.02cm,arrowsize=0.05291667cm 2.0,arrowlength=1.4,arrowinset=0.4]{->}(0.6809375,-2.45)(5.6609373,-0.67)
\psline[linewidth=0.02cm,arrowsize=0.05291667cm 2.0,arrowlength=1.4,arrowinset=0.4]{->}(0.7209375,-0.59)(5.6009374,-1.47)
\end{pspicture} 
}
\caption{The side information edges of the network $\nmas$ in Fig.~\ref{fig:network_example}}
\label{fig:si_edges}
\vspace{-10pt}
\end{figure}

The network $\nma$ is called \emph{solvable} if the source $s$ is able to multicast
$n$ packets to all $k$ sinks simultaneously by using a linear coding scheme 
(see \cite{KoetterMedard2003}). 

\begin{lemma} 
\label{lem:equivalence}
Suppose that $\bA$ is an $m \times k$ binary matrix. 
Then there exists a feasible multicast scheme $\pa$ with $|P| = m$
for $\M$ if and only if the network $\nma$ is solvable. 
\end{lemma} 
\begin{proof}
Assume that there exists a feasible multicast scheme $\pa$ with $|P| = m$
for $\M$. Let $P = \{p_1, \ldots, p_m\}$. 
Then $s$ can multicast $n$ packets $x_1, \ldots, x_n$ to all sinks
in $\nma$ simultaneously using the following coding scheme:
\begin{itemize}
	\item $s$ sends $x_i$ to $s_i$ for every $i \in [n]$, 
	\item $s_i$ $(i \in [n])$ sends $x_i$ to $t_j$ $(j \in [k])$ 
	if they are adjacent,
	\item $s_i$ sends $x_i$ to $u_h$ for every $i \in [n]$ and $h \in [m]$, 
	\item $u_h$ sends $p_h$ to $v_h$ for every $h \in [m]$, 
	\item $v_h$ $(h \in [m])$ sends $p_h$ to $t_j$ $(j \in [k])$ if they are adjacent. 
\end{itemize}
Conversely, assume that the network $\nma$ is solvable. 
By definition, there is a coding scheme so that $s$ can 
multicast $n$ packets to all sinks simultaneously.  
By applying an invertible linear transformation if necessary, 
we can suppose that $s$ sends $x_i$ to $s_i$ for every $i \in [n]$. 
For each $h \in [m]$, let $p_h$ be the packet transmitted on the edge $(u_h, v_h)$. 
Then it is straightforward that $\pa$ where $P = \{p1, \ldots, p_m\}$
is a feasible multicast scheme for $\M$.  
\end{proof} 
\vskip 5pt 

For the instance $\M = \nhj$, for each $j \in [k]$ let
$w_j = n - |H_j|$
denote the number of missing original packets of the client $C_j$.
Let $\weight(\bA[j])$ denotes the number of $1$-entries in
the $j$th column of $\bA$.  
A necessary and sufficient condition for the feasibility of a 
multicast scheme is presented in the following lemma. 
We show that it is possible to satisfy demands from all clients 
if and only if each client receives as many broadcast packets as 
its missing original packets.  

\begin{lemma} 
\label{lem:feasibility}
Suppose that $\bA$ is an $m \times k$ binary matrix. 
Then there exists a feasible multicast scheme $\pa$ with $|P| = m$
for $\M$ if and only if $\weight(\bA[j]) \geq w_j$ for every $j \in [k]$.  
\end{lemma}
\vspace{-3pt}
\begin{proof}
The condition that $\weight(\bA[j]) \geq w_j$ for every $j \in [k]$
is equivalent to the condition that every cut between the source $s$
and a sink in $\nma$ has capacity at least $n$. Due to lack of space, 
we provide a separate proof for this statement in \cite{supplement-DMSI}.  
By the well-known
result from multicast network coding~\cite{KoetterMedard2003}, 
the latter is a necessary and sufficient condition for the solvability of $\nma$.
By Lemma~\ref{lem:equivalence}, we finish the proof.   
\end{proof} 
\vskip 3pt

Lemma~\ref{lem:feasibility} implies that if $\pa$
is a feasible multicast scheme for $\B$ then $|P| \geq \max_j w_j$. 
In Section~\ref{subsec:pas}, we construct a feasible multicast scheme 
that employs precisely $\max_j w_j$ broadcast packets.  
 
\section{Optimal Packet Assignment}
\label{sec:optimality}

In this section, we first describe a feasible multicast scheme $\pas$
for a DMSI instance $\M = \nhj$, and then show that this scheme
obtains the minimum total delay among all feasible multicast schemes
for $\M$. 

\subsection{The Multicast Scheme $\pas$}
\label{subsec:pas}

Relabeling the clients if necessary, we assume that 
\vspace{-5pt}
\begin{equation} 
\label{eq:1}
d_1 \geq d_2 \geq \cdots \geq d_k.  
\vspace{-3pt}
\end{equation} 
\vspace{-5pt}
We consider the multicast scheme $\pas$, where
\begin{equation}
\label{eq:P*} 
m^* = |P^*| = \max_{j \in [k]} w_j,
\vspace{-3pt}
\end{equation} 
\vspace{-7pt}
and $A^* = (a^*_{i,j})$ defined as follows
\begin{equation} 
\label{eq:A*}
a^*_{i,j}=
\begin{cases}
1, & \text{ if } 1 \leq i \leq w_j,\\
0, & \text{ if } w_j < i \leq m^*. 
\end{cases}
\vspace{-3pt}
\end{equation}
We already see an example of such an assignment matrix $\bas$ 
in Fig.~\ref{mat:Astar}, where $\M$ is given in Fig.~\ref{fig:dmsi_example}. 

The broadcast packets of $P^*$ can be obtained as follows. 
By (\ref{eq:P*}) and (\ref{eq:A*}), we have $\weight(\bA[j]) = w_j$
for every $j \in [k]$. Therefore, by Lemma~\ref{lem:feasibility}, 
there exist broadcast packets $p_1, \ldots, p_m$ so that $\pas$ with 
$P^* = \{p_1, \ldots, p_{m^*}\}$ is feasible. Moreover, by the proof of
Lemma~\ref{lem:equivalence}, these broadcast packets can be found
in polynomial time in $n$ and $k$, using the algorithm in \cite{Jaggi-et2005}, 
given that $q \geq k$. For example, let $q = 4$ and $\mathbb{F}_4 = \{0,1,\al,\al^2\}$. 
For $\bas$ given in Fig.~\ref{mat:Astar}, the broadcast packets of 
$P^*$ can be chosen as follows: $p_1 = \al x_3 + x_4 + \al^2x_5 + \al x_6$, 
$p_2 = x_1 + x_2 + \al^2 x_3 + \al x_4 + x_5 + x_6$, 
$p_3 = \al x_1 + \al^2 x_2 + x_3 + \al x_4 + x_5 + \al^2 x_6$, 
$p_4 = x_1 + \al^2 x_3 + \al x_4 + \al^2 x_6$,
$p_5 = \al^2 x_1 + \al x_2 + x_3 + \al x_4 + x_5$.      
   
\subsection{The Optimality of $\pas$}

Now we prove the optimality of $\pas$ in terms of the total delay. 
Let $\pa$ be an arbitrary feasible multicast scheme for $\M$. 
Our goal is to show that $\dt(\bA) \geq \dt(\bA^*)$. 

By Lemma~\ref{lem:feasibility}, the feasibility of $\pa$ implies that 
$\weight(\bA[j]) \geq w_j$ for every $j \in [k]$. Since flipping a $
1$-entry into a $0$-entry does
not increase $\dt(\bA)$, we may assume that $\weight(\bA[j]) = w_j$ 
for every $j \in [k]$. In Lemma~\ref{lem:optimality}, we show that the total
delay of $\pa$ is not smaller than that of $\pas$. 
First, we illustrate the idea of Lemma~\ref{lem:optimality} via an example. 

Consider the DMSI instance $\M$ given in Fig.~\ref{fig:dmsi_example}
together with the delays from the clients given in Fig.~\ref{fig:delay-table}.
Let $\bA$ and $\bas$ be the assignment matrices given in Fig.~\ref{fig:matrices}.  
We now show that $\dt(\bA) \geq \dt(\bas)$ using an algorithmic approach. 
We modify $\bA$ through several steps so that finally, $\bA$ is turned into
$\bas$. Moreover, in every step, $\dt(\bA)$ is never increased. 

\nin{\bf \underline{Step 1}.} We permute the second and the third row of $\bA$. Obviously, 
$\dt(\bA)$ remains unchanged. The matrix now is given in Fig.~\ref{fig:step1}.
We can see that the first columns of $\bA$ and $\bas$ are now the same.

\nin{\bf \underline{Step 2}.} We shift the only $1$-entry in the second column of $\bA$
all the way up to the first row, by swapping $a_{1,2}$ and $a_{3,2}$. 
The matrix now is given in Fig.~\ref{fig:step2}.
As $d_1 \geq d_2$, the broadcast packet $p_1$, which corresponds to the first row of $\bA$, 
still remains to be $d_1$ after the aforementioned swap.
As $a_{3,2}$ is now zero, the delay of the third packet is decreased to $d_4 \leq d_2$. 
These are the only changes in the total delay of $\bA$ after this step. 
Therefore, $\dt(\bA)$ is not increased (in fact, it is decreased by $3$ seconds).
Now the first two rows of $\bA$ and $\bas$ are the same. 

\nin{\bf \underline{Step 3}.}  We first swap $a_{2,3}$ and $a_{4,3}$. The delay of the second 
broadcast packet is still $d_1$ after the swap. The delay of the forth broadcast packet, from $d_3$, 
is now decreased to $d_4 \leq d_3$. Next, we swap the third and the fifth row of
$\bA$. The total delay of $\bA$ is unchanged. 
The matrix now is given in Fig.~\ref{fig:step3}.
The first three rows of $\bA$ and $\bas$ are the same. 
Their forth rows are also identical. 
\vspace{-10pt}
\begin{figure}[H]
\subfloat[$\bA$ after Step 1]{
$
\begin{pmatrix}
1 & 0 & 1 & 1\\
1 & 0 & 0 & 1\\
0 & 1 & 0 & 1\\
0 & 0 & 1 & 1\\
0 & 0 & 1 & 1
\end{pmatrix}
$
\label{fig:step1}
}
\subfloat[$\bA$ after Step 2]{ 
$
\begin{pmatrix}
1 & 1 & 1 & 1\\
1 & 0 & 0 & 1\\
0 & 0 & 0 & 1\\
0 & 0 & 1 & 1\\
0 & 0 & 1 & 1
\end{pmatrix}
$
\label{fig:step2}
}
\subfloat[$\bA$ after Step 3]{ 
$
\begin{pmatrix}
1 & 1 & 1 & 1\\
1 & 0 & 1 & 1\\
0 & 0 & 1 & 1\\
0 & 0 & 0 & 1\\
0 & 0 & 0 & 1
\end{pmatrix}
$
\label{fig:step3}
}
\caption{$\bA$ is turned into $\bas$ in three steps}
\end{figure}
\vspace{-15pt}

\begin{lemma}
\label{lem:optimality}
Let $\bA$ be an $m \times k$ binary matrix where $\weight(\bA[j]) = w_j$ 
for every $j \in [k]$. Then $\dt(\bA) \geq \dt(\bA^*)$. 
\end{lemma}
\vspace{-3pt}
\begin{proof}
Since $\weight(\bA[j]) = w_j$ for all $j \in [k]$, we have 
\vspace{-3pt}
\[
m \geq \max_{j \in [k]} w_j = m^*.
\vspace{-3pt}
\]  
The idea is to repeatedly modify the matrix $\bA$ through $k+1$ steps, so that
at each step, the total delay of $\bA$ is not increased. At the final step, 
$\bA$ is turned into $\bas$. As the total delay never goes up during the
whole process, we conclude that $\dt(\bA) \geq \dt(\bA^*)$. 

Hereafter we say that the two column vectors $\bu \in \fq^m$ and $\bv \in \fq^{m^*}$
are \emph{almost identical} if their first $m^*$ coordinates are identical 
and the last $m-m^*$ coordinates of $\bu$ are all zeros. 

\nin {\bf \underline{Step 1}.} As $\weight(\bA) = \weight(\bas) = w_1$, we can permute
the rows of $\bA$ (if necessary) so that the first columns of $\bA$ and $\bas$ 
are almost identical. 
As permuting rows does not affect the total delay, after 
Step~1, the total delay of $\bA$ remains the same. 

\nin {\bf \underline{Step $j$} $(2 \leq j \leq k)$.} Suppose that up to Step $j-1$, the first
$j-1$ columns of $\bA$ and $\bas$ are almost identical. 
In this step, we modify $\bA$ so that the $j$th columns of $\bA$ and $\bas$
become almost identical. Intuitively, we shift all of the $1$-entries
in the $j$th column of $\bA$ upward as much as we can, and prove that 
during the process, the total delay of $\bA$ is not increased.  
Let
\vspace{-2pt} 
\[
U(j) = \{i \in [m]: \ \exists j' < j \text{ s.t. } a_{i,j'} = 1\},\ L(j) = [m] \setminus U(j).
\]
In words, $U(j)$ denotes the set of upper rows of $\bA$, each of these
contains at least a $1$-entry that is located within the first $j-1$ columns.
Note that $U(j)$ consists of the first $|U(j)|$ rows of $\bA$. 
In opposite, $L(j)$ denotes the set of remaining lower rows of $\bA$, 
where all entries in these rows that are located within the first $j-1$ columns
are zeros. The following modifications to $\bA$ do not increase $\dt(\bA)$
and at the same time, keep the first $j-1$ columns of $\bA$ unchanged. 
\begin{enumerate}
	\item[(M1)] Modify the entries in the $j$th column that are located within 
the first $|U(j)|$ rows. 
As the first $j-1$ columns of $\bA$ and $\bas$ are almost identical, the 
delays (w.r.t $\bA$) of the first $|U(j)|$ broadcast packets are from the set
$\{d_1, \ldots, d_{j-1}\}$. Since $d_j \leq d_{j'}$ for all $j' < j$, any 
change in the $j$th column within the first $|U(j)|$ rows does not 
affect the delays of the corresponding packets.
\item[(M2)] Turn a $1$-entry in the $j$th column that are located within
the last $|L(j)|$ rows into a $0$-entry. The delay of the corresponding broadcast
packet is changed from $d_j$ to $d_{j''}$ for some $j'' > j$. As $d_{j''}
\leq d_j$, the packet delay is not increased. 
\item[(M3)] Permute rows in $L(j)$. It is obvious that permuting rows 
in $\bA$ does not affect $\dt(\bA)$. Moreover, by definition of $U(j)$
and $L(j)$, permuting rows within $L(j)$ does not affect the first 
$j-1$ columns of $\bA$.  
\end{enumerate}
With (M1), (M2), and (M3) in mind, we now apply some modifications to 
$\bA$.
Within the first $|U(j)|$ rows, in the $j$ column of $\bA$, we swap 
pairs of $0$- and $1$-entries such that the $0$-entries are below all the
$1$-entries. Due to (M1),  
$\dt(\bA)$ remains unchanged. We next consider two cases. 
\begin{enumerate}
	\item[(C1)] The $j$th column of $\bA$ has no $1$-entries in the last
	$|L(j)|$ rows. Then we are done for Step $j$ since now the $j$th column of 
	$\bA$ is already almost identical to that of $\bas$. 
	\item[(C2)] The $j$th column of $\bA$ has some $1$-entries in the last
	$|L(j)|$ rows. We now examine only the entries in the $j$th column of $\bA$. 
	\begin{enumerate}
	\item If there are as many $0$-entries in the upper part $U(j)$ as $1$-entries in 
	the lower part $L(j)$ then we can shift the $1$-entries all the way up by 
	applying appropriate entry swaps; 
	and doing so makes the $j$th columns of $\bA$ and $\bas$ almost identical.
	By (M1) and (M2), $\dt(\bA)$ is not increased.  
	\item If there are fewer $0$-entries in the upper part $U(j)$ than $1$-entries
	in the lower part $L(j)$, we first shift as many as we can the $1$-entries from 
	$L(j)$ to $U(j)$; then all entries in $U(j)$ are one. By (M1) and (M2), $\dt(\bA)$ is not increased. 
	Finally, we permute rows in $L(j)$ so that in the $j$th column of $\bA$, 
	the $1$-entries lie above all the $0$-entries. Then the $j$th columns of $\bA$
	and $\bas$ are almost identical. Moreover, by (M3), $\dt(\bA)$ is unchanged. 
	\end{enumerate}
\end{enumerate}
         
\nin{\bf \underline{Step $k+1$}.} The previous $k$ steps guarantee that $\dt(\bA)$ is not increased
and all $k$ columns of $\bA$ are almost identical to that of $\bas$. Therefore,
the last $m^*-m$ rows of $\bA$ are all-zeros. 
In this step, we remove the last $m^*-m$ rows of $\bA$, to turn $\bA$
into $\bas$. Certainly, $\dt(\bA)$ remains unchanged in this step.    
\end{proof} 

\begin{theorem}
\label{thm:main}
The multicast scheme $\pas$ obtains the minimum total delay among
all feasible multicast schemes for $\M = \nhj$. Moreover, 
\vspace{-3pt} 
\begin{equation}
\label{eq:dpas}
\dt\pas = \sum_{j = 1}^k d_j \times \max\big\{0, w_j - \max\{w_j'\}_{0 \leq j' < j}\big\},
\vspace{-3pt} 
\end{equation} 
where we set $w_0 = 0$.
\end{theorem}
\begin{proof}  
The first assertion follows by Lemma~\ref{lem:optimality}. 
To prove that (\ref{eq:dpas}) holds, we show that there are 
\vspace{-3pt} 
\[
\max\big\{0, w_j - \max\{w_j'\}_{0 \leq j' < j}\big\}
\vspace{-3pt} 
\] 
broadcast packets that have delay $d_j$. Obviously, the first 
\vspace{-3pt} 
\[
w_1 = \max\{0, w_1-\max\{w_0\}\}
\vspace{-3pt} 
\]
broadcast packets have delay $d_1$, due to (\ref{eq:packet_delay}) and (\ref{eq:1}). 
By the definition of $\pas$, for each $j > 1$, there are precisely
\vspace{-3pt} 
\[
\max\big\{0, w_j - \max\{w_j'\}_{0 \leq j' < j}\big\}
\vspace{-3pt} 
\]
broadcast packets that are assigned to $C_j$ but to none of the clients 
$C_{j'}$ with $j' < j$. Due to (\ref{eq:packet_delay}) and (\ref{eq:1}), 
these are the only broadcast packets that have delay $d_j$. 
\end{proof} 

\section{Acknowledgment}
The first author thanks Xiaoli Xu for helpful discussions. 
\bibliographystyle{IEEEtran}
\bibliography{Delay-DMSI}

\end{document}